\documentclass[10pt]{article}

\usepackage[top=0.85in,left=1.15in,footskip=0.75in,marginparwidth=1.5in]{geometry}
\usepackage{cite}
\usepackage{nameref,hyperref}
\usepackage{epstopdf}

\usepackage{amsmath, amssymb}
\usepackage{mathenv}
\usepackage[ruled,lined]{algorithm2e}
\usepackage{graphicx}

\usepackage{multirow}
\usepackage{subfig}%
\usepackage{color}
\usepackage{url}
\usepackage{listings}
\usepackage{lscape}

\usepackage{fmtcount}

\newcommand{\F}{\mathbb{F}}

\newcommand{\Z}{\mathbb{Z}}

\newcommand{\Q}{\mathbb{Q}}

\newtheorem{theorem}{Theorem}[section]
\newtheorem{definition}{Definition}[section]
\newtheorem{lemma}[theorem]{Lemma}
\newtheorem{proposition}[theorem]{Proposition}

\newenvironment{proof}[1][Proof]{\begin{trivlist}
\item[\hskip \labelsep {\bfseries #1}]}{\end{trivlist}}
\newenvironment{example}[1][Example]{\begin{trivlist}
\item[\hskip \labelsep {\bfseries #1}]}{\end{trivlist}}

\newcommand{\qed}{\nobreak \ifvmode \relax \else
      \ifdim\lastskip<1.5em \hskip-\lastskip
      \hskip1.5em plus0em minus0.5em \fi \nobreak
      \vrule height0.75em width0.5em depth0.25em\fi}

\begin{document}

\vspace*{0.35in}

\begin{flushleft}
{\Large
\textbf\newline{On the near prime-order MNT curves}
}
\newline
\\

Duc-Phong Le\textsuperscript{1},
Nadia El Mrabet\textsuperscript{2},
Safia Haloui\textsuperscript{3},
Chik How Tan\textsuperscript{4}
%
\\
\bigskip
\bf{1} Institute for Infocomm Research, Singapore
\\
\bf{2} Ecole des Mines de St Etienne, France
\\
\bf{3} Ecole des Mines de St Etienne, France
\\
\bf{4} National University of Singapore, Singapore
\\
\bigskip
* ledp@i2r.a-star.edu.sg

\end{flushleft}




\section*{Abstract}
In their seminar paper, Miyaji, Nakabayashi and Takano introduced the first method to construct families of prime-order elliptic curves with small embedding degrees, namely $k = 3, 4$, and $6$. These curves, so-called MNT curves, were then extended by Scott and Barreto, and also Galbraith, McKee and Valen\c{c}a to near prime-order curves with the same embedding degrees. 
In this paper, we extend the method of Scott and Barreto to introduce an {\em explicit} and {\em simple} algorithm that is able to generate {\em all} families of MNT curves with {\it any} given cofactor. Furthermore, we analyze the number of potential families of these curves that could be obtained for a given embedding degree $k$ and a cofactor $h$.  
We then discuss the generalized Pell equations that allow us to construct particular curves. Finally, we provide statistics of the near prime-order MNT curves.



\section{Introduction}
\label{sec:intro}
Cryptographic pairings were first introduced by Menezes, Okamato and Vanstone in~\cite{MOV91} and Frey and Ruck in \cite{FR94} as a means of attacking discrete logarithm based cryptosystems. 
The authors showed that the discrete logarithm problem on a supersingular elliptic curve could be reduced to the discrete logarithm problem in a finite field through the Weil and Tate pairings. 
Cryptographic pairings on elliptic curves then become a great interest for cryptographic constructions when Joux~\cite{Jou00} introduced the first one-round 3-party Diffie-Hellman key exchange protocol in 2000. 
Since then, pairing-based cryptography has had a huge success with some notable breakthroughs such as the first practical Identity-based Encryption (IBE) scheme~\cite{BF01}. 
Let $E$ be an elliptic curve defined over a finite field $\F_q$ with a subgroup of big prime order $r$. We have:
$$\#E(\F_q) = h \times r,$$

\noindent where $h$ is known as the cofactor. 
In pairing based cryptography, the elliptic curves used have to fulfill a {\em special} property, namely, the embedding degree $k$ is small enough\footnote{The embedding degree is the smallest integer $k$ such that $r$ divides $(q^k-1)$.}. 
This ensures that cryptographic pairings are efficient, that is, computable over the extension finite field. 
An elliptic curve with such a nice property is called a {\em pairing-friendly} elliptic curve. 

In~\cite{MNT01}, Miyaji {\em et al.} introduced the first method that is able to systematically construct ordinary (non-supersingular) elliptic curves of prime order with small embedding degrees $k = 3, 4$ and $6$. 
Their curves, so-called MNT curves, are over fields with large prime characteristic $q$, and the number of points on these curves $E(\F_q)$ is {\it prime}, that is, the cofactor $h = 1$. 
As analyzed in~\cite{PSV2006}, these families of curves are more efficient than supersingular elliptic curves when implementing pairing-based cryptosystems. 
Scott {\em et al.} in~\cite{SB06}, and Galbraith {\em et al.} in~\cite{GMV07} found more ordinary curves of these embedding degrees where the group order $n = \#E(\F_q) = h \times r$ is `nearly prime', that is, $r$ is a prime and $h > 1$ is small.

\subsection{Contributions}
While Galbraith {\em et al.} use the same {\it analytic} technique as in \cite{MNT01} to generate more families of curves with small cofactors $2 \le h \le 5$, Scott {\em et al.}'s method applies the Hasse's bound to generate specific elliptic curves, i.e., actual parameters $q, r, h$ and $D$ (see~\cite[Section 3]{SB06}). 
This paper is an extension from our seminar paper presented at CAI'15~\cite{LET2015}. In this paper, we first extend Scott-Barreto's method~\cite{SB06} to introduce an {\em explicit} and {\em simple} algorithm that allows us to generate families of near prime-order MNT curves. 
Given an embedding degree $k$ and {\it any} cofactor $h_{max} \ge 1$, we will show that our algorithm is able to effectively generate {\em all} families of near prime-order MNT curves having cofactors $h \le h_{max}$. Furthermore, we provide explicit formulas for the number of these families. 
We also analyze the complex multiplication equations of these families of curves and show how to transform these complex multiplication equations into generalized Pell equations. 
Last but not least, we provide some statistics of these near prime-order MNT curves.

\subsection{Organization}
The paper is organized as follows: Section\ref{sec:background} briefly recalls MNT curves, as well as methods to generate MNT curves with small cofactors. Section~\ref{sec:OurMethod} describes our algorithm. We present our families of near prime-order MNT curves in Section~\ref{sec:results}. We also discuss the number of potential families and the Pell equations for some particular cases of MNT curves in this section. Statistics for near prime-order MNT curves are provided in Section~\ref{sec:statisticMNT}. 
Finally, we conclude in Section~\ref{sec:Conclusion}.

\section{Background}\label{sec:background}

Let $E(\F_q)$ be an elliptic curve defined over the finite field $\F_q$, where $q$ is a large prime number. Let $t$ define trace and $r$ be a prime factor of $\#E(\F_q)$. Let $k$ be the embedding degree. $E$ is a pairing-friendly elliptic curve if its embedding degree $k$ is small enough. Balasubramanian and Koblitz~\cite{BK98} pointed out that ordinary elliptic curves generated randomly would have a large embedding degree. Consequently, these curves would not be suitable for efficient computation of a pairing based protocol. Ordinary elliptic curves with small embedding degrees thus require specific constructions.

\subsection{MNT curves}

In~\cite{MNT01}, Miyaji, Nakabayashi, and Takano presented such a construction that yields ordinary elliptic curves with embedding degree $k \in \{3, 4, 6\}$. More particularly, their curves are of prime-order, i.e., the $\rho$-value is $1$ where the value $\rho$ is defined as $\rho = \frac{\log(q)}{\log(r)}$. This is an interest in some applications such as short signatures~\cite{BLS01}.

The families of MNT curves are parametrized by $q$ and $t$ as polynomials in $\Z[x]$ with $\#E(\F_q) = n(x)$. We recall that $n(x)= q(x) + 1 - t(x)$, $n(x) \mid \Phi_k(q(x))$, where $\Phi_k(q(x))$ is the $k$-th cyclotomic polynomial of $q(x)$, and $n(x)$ represents primes in the MNT construction. Their results are summarized in Table~\ref{tab:MNT}. 

\begin{table}[htbp]
  \centering
    \begin{tabular}{|l|c|c|}
     \hline
	$k$ & $q(x)$ & $t(x)$ \\
      \hline
	3 & $12x^2 - 1$ & $-1 \pm 6x$  \\
      \hline
	4 & $x^2 + x + 1$  & $-x$ or $x + 1$  \\
      \hline
	6 & $4x^2 + 1$ & $1 \pm 2x$ \\
      \hline
    \end{tabular}
    \vspace{0.5cm}
  \caption{Parameters for MNT curves~\cite{MNT01}} 
 \label{tab:MNT}
\end{table}

\subsection{Near prime-order MNT curves}
\label{sec:mntcurves}

Let $E(\F_q)$ be a parameterized elliptic curve with cardinality $\# E(\F_q) = n(x)$. We define the cofactor of $E(\F_q)$ as the integer $h$ such that $n(x) = h \times r(x)$, where $r(x)$ is a polynomial representing primes. The original construction of MNT curves gives families of elliptic curves with cofactor $h = 1$. 
Scott-Barreto~\cite{SB06}, and Galbraith-McKee-Valen\c{c}a~\cite{GMV07} extended the MNT idea by allowing small values of the cofactor $h > 1$. This allows us to find many more suitable curves with $\rho \approx 1$ than the original MNT construction. 

\begin{definition}\label{def:nearprime}
Let $E$ be an elliptic curve defined over a finite field $\F_q$. We call $E$ a near prime order curve if its group order $\# E(\F_q)$ is `nearly prime', that is, $\# E(\F_q) = h \times r$ where $r$ is a large prime number and $h$ is a small integer.
\end{definition}

\subsubsection{Scott-Barreto's method}

Let $\Phi_k(x) = d \times r$ for some $x$. Scott-Barreto's method~\cite{SB06} first fixes small integers $h$ and $d$ and then substitutes $r = \Phi_k(t - 1)/d$, where $t = x + 1$ to obtain the following CM equation:

\begin{equation}\label{eq:SB}
 Dm^2 = 4h\frac{\Phi_k(x)}{d} - (x - 1)^2.
\end{equation}

\noindent Scott and Barreto used the fact that $\Phi_k(t - 1) \equiv  0 \pmod r$ (see Proposition~\ref{proposition:1}). As above, the right-hand side of the equation~\eqref{eq:SB} is quadratic, and hence, it can be transformed into a generalized Pell equation by a linear substitution (see~\cite[\S 2]{SB06} for more details). Then, Scott-Barreto found integer solutions to this equation for small enough $D$ (to facilitate the CM algorithm) and arbitrary $m$ with the constraint $4h > d$. The Scott-Barreto's method~\cite{SB06} presented near prime-order MNT elliptic curves with actual parameters, but did not give explicit families of near prime-order MNT elliptic curves.

\begin{proposition}\cite[Proposition 2.4]{FST10}\label{proposition:1}
Let $k$ be a positive integer, $E(\F_q)$ be an elliptic curve defined over $\F_q$ with $\#E(\F_q) = q + 1 - t = hr$, where $r$ is prime, and let $t$ be the trace of $E(\F_q)$. Assume that $r \nmid kq$. Then $E(\F_q)$ has embedding degree $k$ with respect to $r$ if and only if $\Phi_k(q) \equiv 0 \pmod r$, or equivalently, if and only if $\Phi_k(t - 1) \equiv  0 \pmod r$.
\end{proposition}

\subsubsection{Galbraith McKee and Valen\c{c}a's method}
\label{subsec:GMV}

Unlike Scott-Barreto's method, the mathematical analyses in~\cite{GMV07} could lead to explicit families of near prime-order MNT curves. Galbraith {\em et al.}~\cite{GMV07} extended the MNT method~\cite{MNT01} and gave a complete characterization of MNT curves with small cofactors $2 \le h \le 5$. 
As in~\cite{MNT01}, their analysis applies the fact that $\Phi_k(q) \equiv 0 \pmod r$. 
Similar to the method in~\cite{MNT01}, Galbraith {\em et al.} defined $\lambda$ by the equation $\Phi_k(q) = \lambda r$. For example, in the case $k = 6$, they required $\lambda r = \Phi_k(q) = q^2 - q + 1$. 
By using Hasse's bound, $|t| \le 2\sqrt{q}$, they then analyzed and derived possible polynomials $q, t$ from the equation $\Phi_k(q) = \lambda r$.  
Readers are referred to~\cite[Section 3]{GMV07} for a particular analysis in the case, in which the embedding degree is $k = 6$ and the cofactor is $h = 2$. Their results about curves having embedding degrees $k = 3, 4, 6$ with cofactors $2 \le h \le 5$ was summarized in~\cite[Table 3]{GMV07}.

\section{Algorithm}
\label{sec:OurMethod}

In this section, we present an alternative approach to generate {\em explicit} families of ordinary elliptic curves having the embedding degrees $3, 4$, or $6$ and small cofactors. Unlike the analytic approach in~\cite{GMV07}, we obtain families of curves by presenting a very {\em simple} and {\em explicit} algorithm. Given {\em any} cofactor, our analyses also show that this algorithm is able to effectively find all families of near prime-order MNT elliptic curves.

\subsection{Preliminary observations and facts}\label{sec:observations}

Some well-known facts and observations that can be used to find families of curves are noted in this section. 
Similar to Scott-Barreto's method, we use the fact that $\Phi_k(t - 1) \equiv  0 \bmod r$. Consider cyclotomic polynomials corresponding to embedding degrees $k = 3, 4, 6$: 

\begin{equation*}
\Phi_k(t(x) - 1) = t(x)^2 - \varepsilon t(x) + \varepsilon,
\end{equation*}

\noindent and, by setting $t(x) = ax + b$, we have the following equations:

\begin{equation}
\Phi_k(t(x) - 1) = a^2x^2 + a(2b - \varepsilon)x + \Phi_k(b - 1), \label{eq:1}
\end{equation}

\noindent where $\varepsilon =1$ (resp. $2$, $3$) for $k=3$ (resp. $4$, $6$).

\begin{theorem}\label{theo:irreducible}
The quadratic polynomials $\Phi_k(t(x)-1)$ for $k = 3, 4, 6$ are irreducible over the rational field.
\end{theorem}

\begin{proof}
We start with the following lemma that we use later to prove Theorem~\ref{theo:irreducible}.

\begin{lemma}
\label{lem:irr}
Let $f(x)$ be a quadratic irreducible polynomial in $\Q[x]$. If we perform any $\Z$-linear change of variables $x \mapsto ax + b$ for any $a\in \Q \setminus \{0\}$ and $ b \in \Q$, $f(x)$ will still be a quadratic irreducible polynomial in $\Q[x]$. 
\end{lemma}

\begin{proof}
If we assume that $f(ax + b)$ is not irreducible in $\Q[X]$, then as $f(x)$ is a quadratic polynomial it means that $f(ax+b)$ admits a decomposition of the form $f(ax+b)=c(x-c_1)(x-c_2)$, for $c,c_1,c_2 \in \Q$. The values $c_1$ and $c_2$ are rational roots of $f(ax+b)=0$. It is easy to see that $ac_1+b$ and $ac_2+b$ would then be rational roots of $f(x)=0$. \qed
\end{proof}

We now prove Theorem~\ref{theo:irreducible}. 
As the polynomial $\Phi_k(x) = x^2- \varepsilon x + \varepsilon$ is irreducible in $\Q[x]$, according to Lemma~\ref{lem:irr} the polynomial $\Phi_k(t(x)-1)$ is also irreducible in $\Q[x]$. \qed 
\end{proof}

\medskip

Let a triple $(t, r, q)$ parameterize a family of near prime-order MNT curves, and let $h$ be a small cofactor. Let $n(x)$ be a polynomial representing the cardinality of elliptic curves in the family $(t, r, q)$, that is, $n(x) = h \times r(x) = q(x) - t(x) + 1$. By Definition 2.7 in~\cite{FST10}, we have: 
\begin{equation}
  \Phi_k(t(x) - 1) = d \times r(x), \label{eq:41}
\end{equation}

\noindent where $d \in \Z$, and $r(x)$ is a quadratic irreducible polynomial. 
By Hasse's bound, $4q(x) \ge t^2(x)$, we get the inequality: 
\begin{equation}\label{cond1}
4h \geq d
\end{equation}

From Eq.~\eqref{eq:1}, we can see that $d$ is the greatest common divisor (GCD) of the coefficients appearing in this equation. For instance, when $k = 3$, $d$ is the GCD of $\Phi_3(b - 1)$, $a^2$, and $a(2b - 1)$. 
We recall the following well-known Lemma, which can be found in~\cite[Chapter V, \S 6]{Gri07}:

\begin{lemma}
\label{lem:div}
 Let $d$ be prime and $k, n > 0$. If $d$ divides $\Phi_k(n)$, then $d$ does not divide $n$, and either $d$ divides $k$ or $d \equiv 1 \pmod k$.
\end{lemma}

The above lemma points out that if $\Phi_k(n)$ can be factorized by prime factors $d_i$, i.e. $\Phi_k(n) = \prod d_i$, then, either $d_i \mid k$ or $d_i \equiv 1 \pmod k$. 
\begin{example}
In the case of $k = 6$, suppose that $\Phi_6(ax + b') = d \times r(x)$, where $b' = b - 1$. Then $d$ will be the greatest common divisor of $a^2$, $a(2b' + 1)$ and $\Phi_6(b')$. Moreover, either $d | 6$ or $d \equiv 1 \pmod 6$. 
\end{example}

\begin{lemma}\label{lem:dsquarefree}
Given $t(x) = ax + b$, 
if $d$ in Eq.~\eqref{eq:41} does not divide $a$, then $d$ is square free.
\end{lemma}

\begin{proof}
We know that $d \in \Z$, and $d$ is the greatest common divisor of factors of $\Phi_k(t(x)-1)$, i.e. $d$ divides $a^2$, $ 2a(2b-1)$ or $2a(b-1)$ or $2a(2b-3)$ and $\Phi_k(b-1)$ (Eq.~\eqref{eq:1}). 
Suppose that $d$ is not square free, that is $d = p^2\times d'$ with $p$ a prime number greater or equal to 2. 
By Lemma~\ref{lem:div}, $p$ does not divide $(b-1)$ and either $p$ divides $k$ or $p \equiv 1 \pmod k$. 
We also assume that $d$ divides $a^2$, but does not divide $a$, and hence $p^2 \nmid a$, and $p$ is a prime factor of $a$. 

\begin{itemize}

\item $\mathbf{k=3}$: 
As $p$ divides $\Phi_3(b-1) = b^2 - b+1$ and $p$ divides $2b-1$ we have that $p$ divides $(2b-1)+ \Phi_3(b-1)$, {\em i.e.} $p$ divides $b(b-1)$.
We know that $p$ does not divide $(b-1)$, and thus $p$ must divide $b$.

We have $p \mid 2b - 1 = (b - 1) + b$, and $p \mid b$, and hence, $p$ must divide $b - 1$. This contradicts with Lemma~\ref{lem:div}. Thus, $d$ is square free.

\item $\mathbf{k=4}$: We have that $p$ divides $2(b-1)$. But, recall from Lemma~\ref{lem:dsquarefree} that $p$ does not divide $(b-1)$, then $p \mid 2$. However, we can show that $\Phi_4(b - 1) \equiv \{1, 2\} \pmod 4$. 
It is thus impossible to have $d = 2^2 \times d'$ and $d \mid \Phi_4(b - 1)$.

\item $\mathbf{k=6}$:
Likewise, as $p$ divides $\Phi_6(b-1) = b^2 - 3b+3$ and $2b-3$ we have that $p$ divides $(2b-3)+ \Phi_3(b-1) = b(b-1)$. We know that $p$ does not divide $(b-1)$, and so we have $p$ divides $b$.

We have $p$ divides $2b-3$, and $p$ divides $b$, so $p$ must divides $2b-3 + b = 3(b-1)$. Likewise, $p$ does not divide $(b-1)$, and so $p$ divides 3, that is, $d = 3^2 \times d'$. But, by~\cite[Theorem 95]{IntroNumberTheo}, this cannot occur. Thus, $d$ must be square free. \qed
\end{itemize}
\end{proof}

\subsection{The proposed algorithm}

We start this section by presenting the following definition:

\begin{definition}\label{def:algo1}
 Let $t(x),t'(x),r(x),r'(x)$ be polynomials with integer coefficients.
 \begin{itemize}
  \item The $2$-tuple $(t(x),r(x))$ is \textit{deduced} from $(t'(x),r'(x))$ if $(t(x),r(x))=(t'(ux+v),r'(ux+v))$, $u,v\in\mathbb{Z}$, $u\neq 0$.
  \item {\em Equivalence relation}: The $2$-tuple $(t(x),r(x))$ is \textit{equivalent} to $(t'(x),r'(x))$ if both tuples can be deduced from each other, or equivalently,  if $(t(x),r(x)) = (t'(\pm x+v),r'(\pm x+v))$, $v\in\mathbb{Z}$. 
  \item The $2$-tuple $(t(x),r(x))$ is \textit{primitive} if it cannot be deduced from a \\non-equivalent tuple. 
 \end{itemize}
\end{definition}

Algorithm~\ref{algo:getrx} explicitly describes our method. Given an embedding degree $k$ and a cofactor $h_{max}$, Algorithm~\ref{algo:getrx} will output a list of {\em all} possible families of near prime-order MNT curves $(t(x), r(x), q(x))$ with the cofactors $h \le h_{max}$. 

\medskip
\begin{small}
\begin{algorithm}[H]
  \KwIn{An embedding degree $k$, a cofactor $h_{max}$.}
  \KwOut{A list of polynomials $(t(x), r(x), q(x))$.}
  \BlankLine
  \BlankLine
{
$L \leftarrow \{ \}$; $\, \, T \leftarrow \{ \}$ \;
\BlankLine

  \For{$a = -a_{max}$ \em{\KwTo} $a_{max}$}
{
	\For{$b = -b_{max}$ \em{\KwTo} $b_{max}$} {
	  $t(x) \leftarrow ax + b$ \;
	  $f(x) \leftarrow \Phi_k(t(x) - 1)$ \;
	  Let $f(x) = d \cdot r(x)$, where $d \in \Z$ and $r(x)$ is an irreducible quadratic polynomial\;	  
	  \If{$(t(x), r(x))$ couldn’t be deduced from any 2-tuple $(t'(x), r'(x))$ in $T$}{
		$T \leftarrow T + \{(d, t(x), r(x))\}$ \;
        
  \For{$h = \lceil d/4 \rceil$ \em{\KwTo} $h_{max}$} {
 $q(x) \leftarrow h\cdot r(x) + t(x) - 1$ \;
  \If{$q(x)$ is irreducible and $gcd(q(x), r(x): x \in \Z) = 1$} {
	$L \leftarrow L + \{(t(x),r(x),q(x), h)\}$ \;
	
  }
  
}
	  }
	  
	}
}
 
}
  \KwRet{ $L$ }
 \caption{Generate families of near prime-order MNT curves}
 \label{algo:getrx}
\end{algorithm}
\end{small}

\vspace{1cm}

\noindent Basically, given an embedding degree $k$ and a maximum cofactor $h_{max}$, our method works as follows:

\begin{enumerate}
\item Firstly, we set the Frobenius trace polynomial to be $t(x) = ax + b$, for $a \in \Z \setminus \{0\}$ and $b\in \Z$. The possible values of $a, b$ for a given cofactor $h$ are determined by Lemma~\ref{lem:Amax}.
\item Next, we determine $d$ and $r(x)$ thanks to Eq.~\eqref{eq:41}. 
\item If 2-tuple $(t(x), r(x))$ could not be deduced from any 2-tuple in the list $T$, Algorithm~\ref{algo:getrx} adds this tuple into the list. This ensures that the algorithm does not generate any equivalent family of curves (see more details in Section~\ref{sec:results}).
\item Then, for given $t(x), r(x)$ and $d$, we compute the corresponding polynomials $q(x)$ for all cofactors $h \le h_{max}$. 
\end{enumerate}

\noindent Algorithm~\ref{algo:getrx} involves two parameters $a_{max}$ and $b_{max}$. The following section will discuss these values. 

\subsection{Completeness}

The Lemma~\ref{lem:Amax} gives the boundary for the values $a_{max}$, $b_{max}$ in order to find all the possible families of curves. 

\begin{lemma}
\label{lem:Amax}
Given an embedding degree $k$, and a cofactor $h_{max}$, we have $a_{max} = 4h_{max}$, and $b_{max} < a_{max}$.

\end{lemma}

\begin{proof}

We first demonstrate that $a_{max} = 4 h_{max}$. Suppose that $d \mid a^2$, but $d \nmid a$. Then, by Lemma~\ref{lem:dsquarefree}, $d$ must be square free. This is a contradiction, and thus we have $d \mid a$.

Suppose that the algorithm outputs a family of curves with $t(x) = ax + b$, and $a$ is a multiple of $d$, that is, $a = m \times d$. By a $\Z$-linear transformation, we know that this family is equivalent to a family of curves with $t(x) = dx + b$. For the simplest form, the value of the coefficient $a$ of polynomial $t(x)$ should be equal to $d$. Due to the inequality~(\ref{cond1}), the maximum value of $a$, $a_{max} = 4h_{max}$.

Likewise, if $b > a$, we can make a transformation $x \mapsto x + \lfloor b / a \rfloor$, and $b' = b \bmod a$. The value of $b_{max}$ thus should be chosen less than $a_{max}$. \qed

\end{proof}

\section{Families of near prime-order MNT curves}

\label{sec:results}
Algorithm~\ref{algo:getrx} outputs a list of primitive polynomials $\left(t(x), r(x), q(x) \right)$ for all cofactors $h \le h_{max}$. The families of elliptic curves having embedding degrees $k = 3, 4, 6$ and cofactors $h \le 6$ are summarized in Table~\ref{tab:results}. 
Our algorithm executes an {\em exhaustive search} based on the given parameters, and thus it is able to generate {\em all} families of elliptic curves of small embedding degrees 3, 4 and 6. In these tables, we present only families of curves with cofactors $1 \le h \le 6$, but it is worth to note that given any cofactor, families of near prime-order MNT curves can be easily found by adjusting the parameters of Algorithm~\ref{algo:getrx}. 

\begin{theorem}
\label{theo:1}
 Table~\ref{tab:results} gives all families of elliptic curves of the embedding degrees $k = 3, 4, 6$ with different cofactors $1 \le h \le 6$.
\end{theorem}

In comparison to results in ~\cite[Table 3]{GMV07}, note that we provide the {\it primitive} polynomials of $t(x)$, $r(x)$ and $q(x)$ as defined in Definition~\ref{def:algo1}. For example, for $h = 2$, $k = 3$, the family with parameters $q(x) = 8x^2 + 2x + 1$, and $t(x) = -2x$ in~\cite[Table 3]{GMV07} can be deduced to our family with parameters $q(x) = 2x^2 + x + 1$, and $t(x) = -x$. 
For the case of $k = 4$, even though Table 3 in~\cite{GMV07} listed more families than our results, several families of their curves with a given cofactor in~\cite[Table 3]{GMV07} are curves with a higher cofactor. 
For example, when the cofactor is stated to be $h = 2$, with $q(x) = 8x^2 + 6x + 3$, and $t(x) = -2x$, we find that the polynomial $r(x)$ is the following: $r(x) = 2(2x^2 + 2x + 1)$. In this form, $r(x)$ must be divided by 2 before representing primes. Consequently, the cofactor for this family of curves is in fact equal to $4$. This mismatch between the stated cofactor and the real one comes from the fact that in GMV's method the polynomial $r(x)$ does not necessarily represent primes.

We list here the similar cases in Table 3 of~\cite{GMV07} in the case $k = 4$:

\begin{itemize}
 \item $h = 2$: $t = -2x$.
 \item $h = 3$: $t = -2x, t = -10x - 2$, and $t = 10x + 4$.
 \item $h = 4$: $t = -2x, t = -10x - 2$, $t = 10x + 4$, $t = 26x - 4$, and $t = 26x + 6$.
 \item $h = 5$: $t = -2x$, $t = 26x - 4$, $t = 26x + 6$, and $t = -34x - 12$, $t = 34x + 14$.
\end{itemize}

For all these cases, the cofactors are in fact higher than that claimed in~\cite[Table 3]{GMV07}. 
Besides, some families of curves are equivalent by Definition~\ref{def:algo1}. For example, the two families $(t, q) = ((-10x - 1), (60x^2 + 14x + 1))$ and $((10x + 4), (60x^2 + 46x + 9))$ are equivalent. As a result, the number of elliptic curve families is fewer than their claimed number.

\begin{proposition}
\label{propos:1}
Let $q(x), r(x)$ and $t(x)$  be non-zero polynomials  that parameterize a family of near prime-order MNT curves in Table~\ref{tab:results}. 
Then $q'(x) = q(x) - 2t(x) + \varepsilon$, $r(x)$, and $t'(x) =  \varepsilon - t(x)$ represent a family of curves with the same group order $r(x)$ and the same cofactor $h$, where $\varepsilon = 1$ (resp. $2$,~and $3$) for $k = 3$  (resp. $4$,~and $6$).
\end{proposition}

\begin{proof}
 Let $q(x), r(x)$ and $t(x)$ parameterize a family of curves with embedding degrees $k = 3, 4$ or $6$, a small cofactor $h \ge 1$, and let $n(x) = h\cdot r(x)$ represent the number of points on this family of curves. From Eq.~\eqref{eq:1}, we have $\Phi_k(t(x) - 1) = t(x)^2 - \varepsilon t(x) + \varepsilon$. Now, 

\begin{align*}
 \Phi_k(t'(x) - 1) & = \Phi_k(\varepsilon - t(x) - 1) = t(x)^2 -  \varepsilon t(x) +  \varepsilon  \\
 				& = \Phi_k(t(x) - 1).
\end{align*}

Since $r(x) \mid \Phi_k(t(x) - 1)$, we have that $r(x) | \Phi_k(t'(x) - 1)$ and  $q(x) = n(x) + t(x) - 1$. Now,

\begin{align*}
 q'(x) &= q(x) - 2t(x) + \varepsilon = n(x) - t(x) + \varepsilon -1\\
	& = n(x) + t'(x) - 1.
\end{align*}

It is easy to verify that $q'(x)$ is the image of $q(x)$ by a $\Z$-linear transformation of $t(x) \mapsto \varepsilon - t(x)$. According to Lemma~\ref{lem:irr}, since $q(x)$ is irreducible, it follows that $q'(x)$ is also irreducible. If $n'(x) = n(x)$, then the quadratic polynomial $q'(x)$ represents the characteristic of the family of curves.

Now we need to prove that $q'(x)$ and $t'(x)$ satisfy the Hasse's theorem, i.e. $t'(x)^2 \le 4q'(x)$. Suppose that $t(x) = ax + b$, then $t'(x) = -ax - b + 1$. It is clear that the leading coefficient of $q'(x)$ is equal to that of $q(x)$.  Since $h > m/4$, $4q(x)$ would be greater than $t^2(x)$ for some value of $x$. Thus, $q'(x)$ and $t'(x)$ satisfy Hasse's theorem whenever $q(x)$, $t(x)$ involve some big values of $x$. \qed

\end{proof}

\medskip

\begin{landscape}

\begin{table}[h]
\footnotesize
\centering

\begin{tabular}{|p{0.35cm}|p{2.2cm}|p{1.9cm}|p{1.1cm}|p{2cm}|p{1.8cm}|p{1.1cm}|p{2cm}|p{1.8cm}|p{1.1cm}|}
 \hline
 & \multicolumn{3}{c|}{$k = 3$} & \multicolumn{3}{c|}{$k = 4$} & \multicolumn{3}{c|}{$k = 6$} \\
 \hline
  {\bf h} & {\bf q} & {\bf r} & {\bf t}  & {\bf q} & {\bf r} & {\bf t}  & {\bf q} & {\bf r} & {\bf t} \\
 \hline
  \mbox{  1  } & $3x^2 - 1$ & $3x^2 + 3x + 1$ & $-3x - 1$  & $x^2 + x + 1$ & $x^2 + 2x + 2$ & $-x$  & $x^2 + 1$ & $x^2 + x + 1$ & $-x + 1$ \\ 
 \hline
  \multirow{3}{*}{2} & $2x^2 + x + 1$ & $x^2 + x + 1$ & $-x$  & $4x^2 + 2x + 1$ & $2x^2 + 2x + 1$ & $-2x$  & $2x^2 + x + 2$ & $x^2 + x + 1$ & $-x + 1$ \\ 
  & $14x^2 + 3x - 1$ & $7x^2 + 5x + 1$ & $-7x - 2$ &&&   & $6x^2 + 3x + 1$ & $3x^2 + 3x + 1$ & $-3x$  \\
  & $14x^2 + 17x + 4$ & $7x^2 + 5x + 1$ & $7x + 3$ &&& &&&\\  
  \hline
  \multirow{4}{*}{3} & $3x^2 + 2x + 2$ & $x^2 + x + 1$ & $-x$  & $5x^2 + 9x + 9$ & $x^2 + 2x + 2$ & $-x$ & $3x^2 + 2x + 3$ & $x^2 + x + 1$ & $-x + 1$\\
  &&&  & $25x^2 + 15x + 3$ & $5x^2 + 4x + 1$ & $-5x - 1$    & $9x^2 + 6x + 2$ & $3x^2 + 3x + 1$ & $-3x$ \\ 
  &&& & $25x^2 + 25x + 7$ & $5x^2 + 6x + 2$ & $-5x - 2$  & $21x^2 + 8x + 1$ & $7x^2 + 5x + 1$ & $-7x - 1$  \\
  &&& &&& & $21x^2 + 22x + 6$ & $7x^2 + 5x + 1$ & $7x + 4$  \\
 
  \hline
  \multirow{5}{*}{4} & $4x^2 + 3x + 3$ & $x^2 + x + 1$ & $-x$  & $8x^2 + 6x + 3$ & $2x^2 + 2x + 1$ & $-2x$ & $4x^2 + 3x + 4$ & $x^2 + x + 1$ & $-x + 1$ \\ 
  & $12x^2 + 9x + 2$ & $3x^2 + 3x + 1$ & $-3x - 1$  &&& & $28x^2 + 13x + 2$ & $7x^2 + 5x + 1$ & $-7x - 1$ \\
  & $28x^2 + 13x + 1$ & $7x^2 + 5x + 1$ & $-7x - 2$ &&&  & $28x^2 + 27x + 7$ & $7x^2 + 5x + 1$ & $7x + 4$ \\ 
  & $28x^2 + 27x + 6$ & $7x^2 + 5x + 1$ & $7x + 3$ &&&  & $52x^2 + 15x + 1$ & $13x^2 + 7x + 1$ & $-13x - 2$ \\  
  &&& &&&  & $52x^2 + 41x + 8$ & $13x^2 + 7x + 1$ & $13x + 5$ \\
  
  \hline
  \multirow{8}{*}{5} & $5x^2 + 4x + 4$ & $x^2 + x + 1$ & $-x$ & $5x^2 + 9x + 9$ & $x^2 + 2x + 2$ & $-x$  & $5x^2 + 4x + 5$ & $x^2 + x + 1$ & $-x + 1$  \\ 
  & $35x^2 + 18x + 2$ & $7x^2 + 5x + 1$ & $-7x - 2$ 	& $25x^2 + 15x + 3$ & $5x^2 + 4x + 1$ & $-5x - 1$  & $15x^2 + 12x + 4$ & $3x^2 + 3x + 1$ & $-3x$  \\
  
  & $35x^2 + 32x + 7$ & $7x^2 + 5x + 1$ & $7x + 3$ & $25x^2 + 25x + 7$ & $5x^2 + 6x + 2$ & $-5x - 2$  & $35x^2 + 18x + 3$ & $7x^2 + 5x + 1$ & $-7x - 1$  \\  
 
 & $65x^2 + 22x + 1$ & $13x^2 + 7x + 1$ & $-13x - 3$   & $65x^2 + 37x + 5$ & $13x^2 + 10x + 2$ & $-13x - 4$   & $35x^2 + 32x + 8$ & $7x^2 + 5x + 1$ & $7x + 4$ \\
 
  & $65x^2 + 48x + 8$ & $13x^2 + 7x + 1$ & $13x + 4$ 	& $65x^2 + 63x + 15$ & $13x^2 + 10x + 2$ & $13x + 6$   & $65x^2 + 22x + 2$ & $13x^2 + 7x + 1$ & $-13x - 2$ \\
  
  & $95x^2 + 56x + 7$ & $19x^2 + 15x + 3$ & $-19x - 7$ 	& $85x^2 + 23x + 1$ & $17x^2 + 8x + 1$ & $-17x - 3$   & $65x^2 + 48x + 9$ & $13x^2 + 7x + 1$ & $13x + 5$  \\
  
  & $95x^2 + 94x + 22$ & $19x^2 + 15x + 3$ & $19x + 8$ 	& $85x^2 + 57x + 9$ & $17x^2 + 8x + 1$ & $17x + 5$  & $95x^2 + 56x + 8$ & $19x^2 + 5x + 3$ & $-19x - 6$  \\  
  
  &&& &&&   & $95x^2 + 94x + 23$ & $19x^2 + 5x + 3$ & $19x + 9$ \\
 
 \hline
   \multirow{8}{*}{6} 
   & $6x^2 + 5x + 5$ & $x^2 + x + 1$ & $-x$ 	& $12x^2 + 10x + 5$ & $2x^2 + 2x + 1$ & $-2x$ & $6x^2 + 5x + 6$ & $x^2 + x + 1$ & $-x + 1$  \\ 
   & $18x^2 + 15 + 4$ & $3x^2 + 3x + 1$ & $-3x - 1$ & $60x^2 + 26x + 3$ & $10x^2 + 6x + 1$ & $-10x - 2$   & $18x^2 + 15x + 5$ & $3x^2 + 3x + 1$ & $-3x$\\
   & $78x^2 + 29x + 2$ & $13x^2 + 7x + 1$ & $-13x - 3$ & $60x^2 + 46x + 9$ & $10x^2 + 6x + 1$ & $10x + 4$   & $42x^2 + 23x + 4$ & $7x^2 + 5x + 1$ & $-7x - 1$ \\ 
   & $78x^2 + 55x + 9$ & $13x^2 + 7x + 1$ & $13x + 4$ & $102x^2 + 31x + 2$ & $17x^2 + 8x + 1$ & $-17x - 3$   & $42x^2 + 37x + 9$ & $7x^2 + 5x + 1$ & $7x + 4$ \\    
   & $114x^2 + 71x + 10$ & $19x^2 + 15x + 3$ & $-19x - 7$ & $102x^2 + 65x + 10$ & $17x^2 + 8x + 1$ & $17x + 5$   & $78x^2 + 29x + 3$ & $13x^2 + 7x + 1$ & $-13x - 2$ \\
   
   & $114x^2 + 109x + 25$ & $19x^2 + 15x + 3$ & $19x + 8$ &&& & $78x^2 + 55x + 10$ & $13x^2 + 7x + 1$ & $13x + 5$ \\
   & $126x^2 + 33x + 1$ & $21x^2 + 9x + 1$ & $-21x - 4$ &&& &&& \\     
   & $126x^2 + 75x + 10$ & $21x^2 + 9x + 1$ & $21x + 5$ &&& &&& \\  
  \hline

  \end{tabular}  
  \hspace{.5cm}%
 
\vspace{0.5cm}

\caption{Valid $q, r, t$ corresponding to the embedding degrees $k = 3, 4, 6$}
\label{tab:results}
\end{table}

\end{landscape}

\subsection{The number of potential families}

Let $k\in \{ 3,4,6\}$. The families with parameters $(t(x),r(x),q(x))$ of near prime-order MNT curves built by Algorithm~\ref{algo:getrx} are characterized by the following properties :
\begin{enumerate}
 \item[(1)] $t(x)=ax+b$, $a,b\in\mathbb{Z}$, $a\neq 0$, 
 \item[(2)] $r(x)$ is the $\mathbb{Z}$-irreducible polynomial such that $\Phi_k(t(x)-1)=d \times r(x)$ for some $d\in\mathbb{N}$,
 \item[(3)] $q(x)=hr(x)-t(x)-1$, where $h$ is a positive integer satisfying $4h\geq d$.
\end{enumerate}
If $h$ is a fixed positive integer, we see that the number of such families is equal to 
$$\sum_{d=1}^{4h}N_d,$$
where $N_d$ is the number of primitives classes having a representation $(t(x),r(x))$ satisfying properties~(1) and~(2). The purpose of this section is to give an explicit formula for the value of $N_d$. 
Let us recall equations~\eqref{eq:1} and~\eqref{eq:41} in Section~\ref{sec:observations}:

$$\Phi_k(t(x)-1)=a^2x^2+a(2b-\varepsilon )x+\Phi_k (b-1),$$
where $\varepsilon =1$ (resp. $2$, $3$) for $k=3$ (resp. $4$, $6$). The integer $d$, satisfying the equation
$$\Phi_k(t(x)-1)=d \times r(x)$$
is the gcd of $a^2$, $a(2b-\varepsilon )$ and $\Phi_k (b-1)$. It is proved in Lemmas~\ref{lem:div} and~\ref{lem:dsquarefree} that $d$ divides $a$, so $d$ is also the gcd of $a$ and $\Phi_k (b-1)$. 
Moreover, it is easy to see that $(t(x),r(x))$ is always deduced from the couple $(dx+b,\Phi_k(dx+b-1)/d)$, so any primitive couple must have $a=d$, or
equivalently, $a\vert \Phi_k (b-1)$.

\begin{lemma}\label{Nd}
 Let $d$ be a fixed positive integer. We have
 $$N_d=\# \{ b\mod d,\quad d\mid \Phi_k(b-1)\}.$$
\end{lemma}
\begin{proof}
 Taking into account the discussion above, it is easy to check that we have a bijection between the set of primitives classes having a representation $(t(x),r(x))$ satisfying (1) and (2) and $ \{ b\mod d,\quad d\mid \Phi_k(b-1)\}$ which is given by $(t(x), r(x))\mapsto b \mod d$.
\end{proof}

\begin{proposition}
  Let $d$ be a fixed positive integer and write $d=p^{u_0}q_1^{u_1}\dots q_s^{u_s}$, where $p$ is the biggest prime factor of $k$ (so $p=2$ or $3$), $q_1, \dots, q_s$ are distinct primes and distinct from $p$, and $u_0, \dots, u_s\in\mathbb{N}$. We have that
  $$
  N_d=
  \left\{
  \begin{array}{rl}
	1 & \mbox{if } d = 1 \mbox{ or } d = p, \\
	2^s & \mbox{if } q_i=1\mod k, i=1,\dots ,s,\mbox{ and } u_0\leq 1, \\
	0 & \mbox{otherwise.}
  \end{array}
  \right.
  $$
\end{proposition}
\begin{proof}
By Lemma~\ref{Nd}, we are reduced to find the number of elements in the set $\{ a\mod d,\quad d\mid \Phi_k(a)\}$. We remark that it is trivial that $N_1 = 1$. For the higher value of $d$, we will make use of the results from~\cite{IntroNumberTheo}.

 $\bullet$ Case $d=p$: Let $k=mp^e$, $p\nmid m$ (so $m=1$ or $2$). There is exactly one $a\in\mathbb{Z}/p\mathbb{Z}$ such that $\mbox{ord}_p(a)=m$, so by \cite[Theorem 95]{IntroNumberTheo}, we have $N_p=1$.
 
 $\bullet$ Case $d=p^{u_0}$: We have $N_{p^{u_0}}=0$ by~\cite[Theorem 95]{IntroNumberTheo}. 
 
 $\bullet$ Case $d=q$, $q$ prime distinct from $p$: By \cite[Theorem 95]{IntroNumberTheo}, $N_q$ is the number of $a\in\mathbb{Z}/q\mathbb{Z}$ such that $\mbox{ord}_p(a)=k$, which is equal to $\varphi(k)=2$ if $q=1\mod k$ and $0$ otherwise.
 
 $\bullet$ Case $d=q^u$, $q$ prime distinct from $p$: By Hansel's Lemma, we have $N_{q^u}=N_q$ (Hansel's Lemma applies because the only prime factor of the discriminant of $\Phi_k$ is $p$).
 
 $\bullet$ General case: By the Chinese Remainder Theorem, we have $N_{p^{u_0}q_1^{u_1}\dots q_s^{u_s}}=N_{p^{u_0}} N_{q_1^{u_1}}\dots N_{q_s^{u_s}}$. \qed

\end{proof}

\subsection{Solving the Pell Equations}
\label{sec:SolvingPellEquations}

Solving the Pell equations for MNT curves was studied in papers~\cite{KT08} and~\cite{FK13}. The authors proved that MNT curves are {\em sparse}, that is, Pell equations admit only a few solutions. 
In this section, we extend their ideas to solve the Pell equations for near prime-order MNT curves. 

Let $t(x) = ax + b$, $\Phi_k(t(x) - 1) = d \cdot r(x)$, where $k = 3, 4, 6$ and $\#E(\F_q) = h\cdot r(x)$. 
Let $\varepsilon = 1$ (resp. $2$,~and $3$) when $k = 3$  (resp. $4$,~and $6$). 
In order to remove the linear term in the CM equation $Dm^2 = 4q(x) - t^2(x)$ of the near prime order MNT curves, we substitute $x = (y - a_k)/n$, where $n = a(4h - d)$, and $a_k = 2h(2b - \varepsilon) - (b - 2)d$ for $k = 3, 4$, or $6$. 
The CM equation can be transformed into a generalized Pell equation of the form:

\begin{equation}
\label{eq:pell}
 y^2 - gm^2 = f_k,
\end{equation}

\noindent where $g = d(4h - d)D$ and $f_k = a_k^2 - ((4h - d)b)^2 + 4(4h - d)(b - 1)(\varepsilon h - d)$. 

By fixing $a = 1$ and $b = 1$, one can get the values $a_k$ and $f_k$ as analyzed in~\cite[Section 2]{SB06}. Note that there is a typo in the value of $f_k$ in \cite[Section 2]{SB06}. Indeed, $f_k$ must be set to $a_k^2 - b^2$ instead of $a_k^2 + b^2$. The following section illustrates our method for $k=6$ and $h=4$ as follows.

\subsubsection{Case $k = 6$ and $h = 4$}

Elliptic curves with the cofactor $h = 4$ may be put in the form $x^2 + y^2 = 1 + dx^2y^2$ with $d$ a non-square integer. Such curves called Edwards curves were introduced to cryptography by Bernstein and Lange~\cite{BL07}. They showed that the addition law on Edwards curves is faster than all previously known formulas. Edwards curves were later extended to the twisted Edwards curves in~\cite{BBJ+08}. Readers also can see~\cite{ALNR10}\cite{LT14} for efficient algorithms to compute pairings on Edwards curves. We give in this section some facts to solve Pell equation for Edwards curves with embedding degree $k = 6$. By using Eq.~\eqref{eq:pell}, we obtain the following Pell equations:

\begin{align}
 y_1^2 -g_1m^2 & = -176, \label{eq:pell1}\\
 y_2^2 - g_2m^2 & = -80, \label{eq:pell2}\\
 y_3^2 - g_3m^2 & = -80, \label{eq:pell21}\\
 y_4^2 -g_4m^2 & = 16 \label{eq:pell3}, \\
 y_5^2 -g_5m^2 & = 16 \label{eq:pell31},
\end{align}

\noindent where $y_i = (x - a_i)/b_i$, $g_i = b_iD$, for $i \in [1, 5]$, and 

\begin{center}
\begin{tabular}{c c c c c}
$a_1 = -7$, & $a_2 = -19$, & $a_3 = -26$, & $a_4 = -4$, &  $a_5 = -17$, \\
$b_1 = 15$, & $b_2 = 63$, & $b_3 = 63$, & $b_4 = 39$, &$ b_5 = 39$.\\
\end{tabular}
\end{center}

\noindent Karabina and Teske~\cite[Lemma 1]{KT08} showed that if $4 \mid f_k$, then the set of solutions to $y^2 - gm^2 = f_k$ does not contain any {\em ambiguous} class, i.e., there exists no primitive solution $\alpha = y + v\sqrt{g}$ such that $\alpha$ and its {\em conjugate} $\alpha' = y - v\sqrt{g}$ are in the same class. Consequently, equations~\eqref{eq:pell1}--\eqref{eq:pell31} do not have any solution that contains an ambiguous class. 

If equations~\eqref{eq:pell1}--\eqref{eq:pell31} have solutions with $y_i \equiv -a_i \bmod{b_i}$, and a fixed positive square-free integer $g_i$ relatively prime to $b_i$, for $1 \le i \le 5$, then triple $t, r, q$ in Table~\ref{tab:results} with $k = 6$ and $h = 4$ represent a family of pairing-friendly Edwards curves with embedding degree 6.

\section{Statistics of near prime-order MNT curves}\label{sec:statisticMNT}
In~\cite{ULS2012}, Jim\'enez Urroz, Luca and Shparlinski provided statistics of MNT curves in the case $k = 6$. In this section, we generalize their arguments to the near prime-order MNT curves.

\begin{theorem}[\cite{ULS2012}, Theorem 8]
Let $E(z)$ be the number of MNT curves with $k = 6$ and co-factor $h=1$ having CM discriminant less than $z$. Then, assuming the Generalized Bateman-Horn Conjecture, the lower bound

$$E(z)\geq (\mathfrak{S}_0+o(1))\frac{\sqrt{z}}{\log z},$$
\noindent holds as $z \rightarrow \infty$, where $\mathfrak{S}_0\simeq 0.237615$.
\end{theorem}

\subsection{Assumptions}
Let $D$ denote the CM discriminant. We first rewrite Eq.~\eqref{eq:pell} in the following form:
\begin{equation}
\Delta (x)=Dum^2,
\label{eq:delta}
\end{equation}

\noindent where $\Delta (x) = (w_0x+w_1)^2 + w_2$, $w_0 = a(4h-d)$, $w_1 = b(4h - d) – 2(\varepsilon h - d)$, $w_2 = 4h(4 - \varepsilon) (\varepsilon h - d)$, $u = d(4h-d)$, and parameters $a, b, h, d$ and $\epsilon$ are defined as in Section~\ref{sec:SolvingPellEquations}. Note that these parameters were straightforwardly deduced from Eq.~\eqref{eq:pell}. 

In order to fulfil the conditions of the \textit{generalized Bateman-Horn conjecture} given in \cite[Section 3.4]{ULS2012}, 
we assume that the products $r(n)q(n)\Delta (n)$, $n\in\mathbb{Z}$ have no fixed prime divisor. 
We have the following lemma whose proof is straightforward.

\begin{lemma}\label{ConditionsDelta}
Under the above assumptions, for any prime $p$ and any integer $\beta$ such that $\Delta (\beta )=0\mod p^2$, we have 
 $\Delta' (\beta )\neq 0\mod p$.
\end{lemma}

\begin{proof}
If $p$ is an odd prime that doesn’t divide the leading coefficient of $\Delta (x)$, then the equation $\Delta (x) = 0 \mod p^2$ is equivalent to $(\Delta' (x))^2-\mbox{Disc}(\Delta (x))=0\mod p^2$, and we see that if this equation has a solution $\beta$ such that  $\Delta' (\beta )= 0\mod p$, then $p^2$ must divide $\mbox{Disc}(\Delta (x))$.
 
If $p$ is an odd prime dividing the leading coefficient of $\Delta (x)$.
Under these conditions, it is easy to check that if $\Delta (x)$ has a root modulo $p$, then either $\Delta (x)$ is identically zero modulo $p$ (which is excluded, since $\Delta (x)$ is irreducible over $\mathbb{Z}[x]$), or $\Delta' (x)$ has no root modulo $p$.
 
Finally, we can easily check that $\Delta (x)$ has no root modulo $2$. \qed
\end{proof}

\noindent Similar to Jim\'enez Urroz {\em et al.}'s analysis, we proceed in two steps: 
\begin{itemize}
\item in the first step, let $m$ be an fixed integer, $B(m)$ be the set of residues module $m^2$, and let $\beta \in B(m)$. We seek to estimate $E_{m,\beta}(z)$, the number of positive integers of the form $n = \beta + km^2$ that satisfies the following conditions:
\begin{itemize}
 \item[ (1)] $q(n)$ is prime,
 \item[ (2)] $r(n)$ is prime,
 \item[ (3)] $\Delta (n)/m^2$ is a square-free integer $\leq z$, for some positive integer $m$.\\
\end{itemize}

\item the second step will give a lower bound of the sum of all the $E_{m,\beta}(z)$. 
\end{itemize}

\subsection{Preparations} 
If $n = \beta + km^2$ (that is, $n \equiv \beta \mod m^2$) satisfying the above conditions (1)--(3), then from Eq.~\eqref{eq:delta}, the class of $n$ modulo $m^2$ is an element of the set $B(m)$ of solutions of the following equation:
\begin{eqnarray}\label{CM_mod_m2}
\Delta(n) = 0 \mod m^2
\end{eqnarray}
From the condition $D \leq z$, we get 
$$w_0^2k^2m^4\leq \Delta (n)\leq zum^2$$
so $k \leq \sqrt{zu}/(w_0m)$. If the conditions (1)--(3) were independent events, the number of positive integers of the form 
$n = \beta + km^2$ satisfying these conditions would behave like 
$$F_{m,\beta}(z)=\frac{\sqrt{zu}}{\zeta (2)w_0m(\log (zm^2))^2}.$$

However, these events are actually not  independent, so the estimate $F_{m,\beta}(z)$ needs to be corrected by a constant which takes in account the local behaviors of $q(n)$, $r(n)$, $\Delta (n)/m^2$. This gives:

\begin{eqnarray*}
E_{m,\beta}(z) & \sim & \prod_p\left(1-\frac{N_{m,p}}{p^2}\right)\left(1-\frac{1}{p}\right)^{-2}\frac{\sqrt{zu}}{w_0m(\log (zm^2))^2},
\end{eqnarray*}
where $N_{m,p}$ is the number of solutions in the arithmetic progression $n=\beta\mod m^2$ of the congruence 

\begin{eqnarray}\label{BH/m2_mod_p2}
\left(q(n)r(n)\right)^2\frac{\Delta (n)}{m^2}=0\mod p^2.
\end{eqnarray}

Let $C_p$ be the number of solutions of the congruence:

\begin{eqnarray}\label{BH_mod_p2}
\left(q(n)r(n)\right)^2\Delta (n)=0\mod p^2.
\end{eqnarray}

\begin{description}

\item[1. ] If $p\nmid m$, then we see that $N_{m,p}=C_p$. The $C_p$'s can be computed by using the fact that almost all primes $p$ divide at most one of $q(n)$, $r(n)$ and $\Delta (n)$, and for these primes, we can get the solutions of (\ref{BH_mod_p2}) by counting separately the roots of $q(n)$, $r(n)$ and $\Delta (n)$ modulo $p$. In particular, we have $C_p=O(p)$. 
The remaining primes should be treated ``by hand''
(such primes must divide $2u$ or the resultant of two polynomials in $\{ q(n), r(n), \Delta (n)\}$). \\

\item[2. ] If $p\mid m$, then we have two possibilities: 
\begin{itemize}
\item[(i)] either $p\mid q(\beta)$ (resp. $p\mid r(\beta)$) and then $q(\beta +km^2)\in\mathbb{Z}[k]$ 
(resp. $r(\beta +km^2)$) does not take any prime value (this can happen only for a finite number of primes, namely, the primes which
divide $\mbox{Res}\left(q(x), \Delta (x)\right)\mbox{Res}\left(r(x), \Delta (x)\right)$);
\item[(ii)] or $p\nmid q(\beta )r(\beta )$, so 
$p\nmid q(n)r(n)$ for any $n=\beta +km^2$ and the Hansel's Lemma and Lemma~\ref{ConditionsDelta} ensure that there exists an unique solution modulo $p^2$ to the equation 
$$\frac{\Delta (\beta +km^2)}{m^2}=0\mod p^2,$$
and therefore, we have $N_{m,p}=1$.

\end{itemize}

\end{description}

\subsection{Lower bound on near prime-order MNT curves}

Let $\beta\in B'_m=\left\{\beta\in B_m\:\vert\: \forall p \mbox{ dividing } m,\: q(\beta )r(\beta )\neq 0\mod p\right\}$, and define $\rho (m) = \# B'_m$. The following lemma gives some basic properties of the function $\rho$:

\begin{lemma}\label{ProprietesRho}
Let $\beta\in B'_m=\left\{\beta\in B_m\:\vert\: \forall p \mbox{ dividing } m,\: q(\beta )r(\beta )\neq 0\mod p\right\} $. The function $\rho (m)= \# B'_m$ verifies that:
 \begin{enumerate}
  \item The function $\rho$ is multiplicative.
  \item For almost all primes, we have $\rho (p^e)=\rho (p)$, $\forall e\in \mathbb{N}^*$.
  \item  Let $w'_2$ be the product of the odd prime divisors of the non-square part of $w_2$ and $\ell =2\varphi (w'_2)$. Then there exist exactly $\ell$ classes
         $c_1,\dots ,c_\ell$ modulo $8w'_2$ (which can be explicitly computed) such that for almost all primes, we have
         $$\rho (p)=
         \left\{
         \begin{array}{rl}
         2 & \mbox{if }p=c_i\mod 8w'_2,\mbox{ for some }i\in\{1,\dots ,\ell\},\\
         0 & \mbox{otherwise.}
         \end{array}
         \right.
         $$
 \end{enumerate}
\end{lemma}

\begin{proof}
 Point (1) follows directly from the Chinese Remainder Theorem. For almost all primes $p$ and all $e\in \mathbb{N}$, $\rho (p^e)$ is just the number of solutions to (\ref{CM_mod_m2}) with $m=p^e$, so point (2) follows from the Hansel Lemma. As for point (3), notice that for almost all primes $p$, the number of solutions to $\Delta(x)=0\mod p$ is $2$ if $-w_2$ is a square modulo $p$, and $0$ otherwise. We conclude the proof by using the Law of Quadratic Reciprocity.
 \qed 
\end{proof}

Let $E(z)$ be the number of near prime order MNT curves having CM discriminant $D < z$ and let $M$ be an integer. We have

\begin{gather}
\begin{aligned}
E(z) \ge \sum_{m\leq M}\sum_{\beta\in B_m}E_{m,\beta}(z) & \ge \mathfrak{S}_1\frac{\sqrt{zu}}{w_0}\sum_{m\leq M}\frac{f(m)}{m(\log (zm^2))^2} \end{aligned}
\end{gather}

\noindent where
$$\mathfrak{S}_1=\prod_{p}\left(1-\frac{C_{p}}{p^2}\right)\left(1-\frac{1}{p}\right)^{-2},$$
\noindent and 
$$f(m)=\prod_{p\mid m}\left(1-\frac{1}{p}\right)^2\left(1-\frac{C_{p}}{p^2}\right)^{-1}\rho (p),$$
With the notation of Lemma~\ref{ProprietesRho}, we have
\begin{eqnarray}\label{Proprietef(p)}
 f(p)=
 \left\{
 \begin{array}{rl}
  2+O(1/p) & \mbox{if }p=c_i\mod 8w'_2,\mbox{ for some }i\in\{1,\dots ,\ell\},\\
  0 & \mbox{otherwise.}
 \end{array}
 \right.
\end{eqnarray}
Now, let 
$$S(t)=\sum_{m\leq t}\mu (m)^2\frac{f(m)}{m},$$
where $\mu$ is the M\"obius function. By using (\ref{Proprietef(p)}), we can follow exactly the method given at the beginning of the proof of \cite{ULS2012}, Theorem 8.
This gives
\begin{eqnarray}\label{EstimationS(T)}
S(t)=\mathfrak{S}_2\log (t)+O(1),
\end{eqnarray}
where 
\begin{eqnarray}\label{S2}
\mathfrak{S}_2=\prod_{p}\left(1-\frac{1}{p}\right)\left(1+\frac{f(p)}{p}\right)
\end{eqnarray}
Still following the proof of \cite{ULS2012}, Theorem 8, we use the partial summation (see \cite{IK2004}, Section 1.5) and Equation~\ref{EstimationS(T)} we found that

\begin{eqnarray}\label{MinorationFinale}
\sum_{m\leq M}\mu (m)^2\frac{f(m)}{m(\log (zm^2))^2}\geq \frac{\mathfrak{S}_2}{2}\left( \frac{-1}{\log (zM^2)}+\frac{1}{\log (z)}\right) +O\left(\frac{1}{\log (z)^2}\right) \\ \nonumber 
\end{eqnarray}
Taking $M=z^{1/2}$ in (\ref{MinorationFinale}), we get the following theorem. %

\begin{theorem}

Let $E(z)$ be the number of curves in the family with parameters $(t,r,q)$ polynomials in $x$,  an embedding degree $k$ and a co-factor $h$ having discriminant $D$ less than $z$. Given:
\begin{eqnarray*}
\mathfrak{S}_1=\prod_{p}\left(1-\frac{C_{p}}{p^2}\right)\left(1-\frac{1}{p}\right)^{-2},
\end{eqnarray*}
 and 
 \begin{eqnarray*}\label{S2}
\mathfrak{S}_2=\prod_{p}\left(1-\frac{1}{p}\right)\left(1+\frac{f(p)}{p}\right).
\end{eqnarray*}
 
\noindent Then, the lower bound 
 $$E(z)\geq \left(\mathfrak{S}_0+o(1)\right)\frac{\sqrt{z}}{\log (z)},$$
holds as $z \rightarrow \infty$, where $\mathfrak{S}_0=\frac{\sqrt{u}}{4w_0}\mathfrak{S}_1\mathfrak{S}_2$.
\end{theorem}

\section{Conclusion}
\label{sec:Conclusion}
In this paper, we first extended Scott-Barreto's method and presented an explicit and efficient algorithm that is able to generate all families of the near prime-order MNT curves, given an embedding degree $k$ and a cofactor $h$. Furthermore, we provided explicit formulas for the number of these families. Then, we analyzed the generalized Pell equations of these curves. Finally, we gave statistics of the near prime-order MNT curves.

\bibliographystyle{abbrv}
\bibliography{xMNT}

\end{document}